\documentclass{amsart}

\usepackage{graphicx}
\usepackage{mathtools}
\usepackage{hyperref}
\usepackage{url}

\newtheorem{theorem}{Theorem}[section]
\newtheorem{lemma}[theorem]{Lemma}
\newtheorem{conjecture}[theorem]{Conjecture}
\newtheorem{corollary}[theorem]{Corollary}
\newtheorem{observation}[theorem]{Observation}
\newtheorem{fact}[theorem]{Fact}

\theoremstyle{definition}
\newtheorem{definition}[theorem]{Definition}

\theoremstyle{remark}

\numberwithin{equation}{section}

\begin{document}

\title{On the Complexity of Slide-and-Merge Games}

\author{Ahmed Abdelkader}

\author{Aditya Acharya}
\author{Philip Dasler}




\keywords{Games and Puzzles, Computational Complexity}

\begin{abstract}
We study the complexity of a particular class of board games, which we call `slide and merge' games. Namely, we consider \emph{2048} and \emph{Threes!}, which are among the most popular games of their type. In both games, the player is required to slide all rows or columns of the board in one direction to create a high value tile by merging pairs of equal tiles into one with the sum of their values. This combines features from both block pushing and tile matching puzzles, like \emph{Push} and \emph{Bejeweled}, respectively. We define a number of natural decision problems on a suitable generalization of these games and prove NP-hardness for \emph{2048} by reducing from \verb|3-SAT|. Finally, we discuss the adaptation of our reduction to \emph{Threes!} and conjecture a similar result.
\end{abstract}

\maketitle


.

\section{Introduction}

The original \emph{2048} game is played on a $4\times4$ grid, with numbered tiles that slide along rows or columns when the player makes a move in any of the four directions $\{\leftarrow,\rightarrow,\uparrow,\downarrow\}$. Tiles slide as far as possible in the chosen direction until they are stopped by either another tile or the edge of the board. If two tiles with the same value collide, they merge into a tile with twice the value. Note that the resulting tile cannot merge further in the same move. After each turn, a \emph{2} or \emph{2} tiles is generated in any of the empty cells. The player wins when a \emph{2048} tile is created, hence the name of the game. Otherwise, the player loses when the board is full and no merges can be performed. Figure \ref{fig:2048_move} shows the result of making a $\downarrow$.

\begin{figure}[htb]
\includegraphics[width=0.65\textwidth]{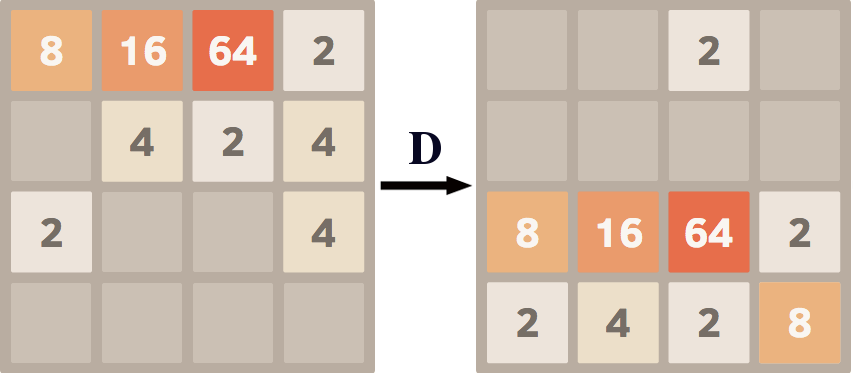}
\caption{The result of taking the \textbf{Down} action.  Notice that the two 4s on the right have merged and a new 2 has been randomly inserted.}\label{fig:2048_move}
\end{figure}

We generalize \emph{2048} to allow a rectangular grid of arbitrary size, with the same tile dynamics. We assume we are given the complete board configuration with the tiles and their values, and no new tiles are generated during the course of the game. Next, we define a number of natural decision problems related to the generalized game.
\begin{itemize}
\item Given an $m \times n$ board with some configuration of tiles, is it possible to create a tile of a certain value $v = 2^k$?
\item Given an $m \times n$ board with some configuration of tiles, is it possible to merge two particular tiles?
\item What is the answer to the above decision problems if the player is limited to making at most $k$ moves?
\end{itemize}
In particular we restrict ourselves to the first decision problem, i.e. we ask if producing a tile of value $2^{k}$ is possible. To formalize we ask the following yes-no question
\begin{definition}(\verb|2048-GAME|)
Given a starting configuration on a $n\times m$ board, is it possible to obtain a tile of value $2^{k}$ ? 
\end{definition}

The main theorem we prove in this paper is the following.

\begin{theorem}\verb|2048-GAME| is NP-hard.\end{theorem}


We develop a reduction from \verb|3-SAT|, with $m$ clauses and $n$ variables, to an instances of \verb|2048-GAME| having a number of tiles that is polynomial in both $m$ and $n$ with the objective of obtaining a $2048$ tile.

In \emph{Threes!}, tiles of value $1$ and $2$ merge to produce $3$. There on, a tile of value $3\cdot2^{k}$ merges with another of the same value to produce a tile of value $3\cdot2^{k+1}$. Another important distinction, as far as tile dynamics are concerned, is that the rows and columns in \emph{Threes!} slide one step at a time, unlike \emph{2048} where the tiles slide \textit{all the way}. We define similar decision problems on an $m\times n$ board, with no random tiles being generated.

\begin{definition}(\verb|Threes-GAME|)
Given a starting configuration on a $m\times n$ board, is it possible to obtain a tile of value $3\cdot k$? 
\end{definition}

We anticipate the following result using a similar reduction.
\begin{conjecture}\verb|Threes-GAME| is NP-Hard.\end{conjecture}

\section{Related Work}

Some of the older games involving matching for example, \emph{Bejeweled, Candy
Crush Saga} etc. have been proven to be hard. \cite{match3}
showed that they are NP-Hard for a number of questions like determining
if the player can reach a certain score or play for a certain number of
turns. The reduction is from \verb|1-in-3 Positive 3SAT| where a certain
set of tiles cannot be merged unless a prespecified shift is applied to
their neighborhood. The amount of shift is related to the number of
clauses in the \verb|3SAT| instance, which in turn are encoded using
a pattern of tiles that may be cleared by an appropriate assignment
of variables. Such assignments are manifested by the choice of which
tiles to merge, where the player is given two options for each variable
that correspond naturally to possible truth assignments of that variable.

Another class of puzzles having similar tile dynamics, are block
pushing puzzles. \emph{Push-1 } is a game where a robot has to fill
certain target cells by pushing blocks into them. In \emph{PushPush }, a variant of  \emph{Push-1 }, blocks slide all the way
until they hit another block or a wall. \cite{pushpush} proved these problems
to be NP-Hard.  This game was already known to be hard in 3D \cite{pushpush3D} by
a reduction from \verb|SAT|. The result in \cite{pushpush} embeds the graph
obtained from \cite{pushpush3D} into the plane by inserting special gadgets
to resolve path intersections.

One recent game that went viral earlier in 2014 is \emph{2048}. In this game,
two tiles of equal value merge into a single tile of double the value. The
goal is to create a single high value tile. This may be achieved by
sliding all rows or all columns in one direction. Other than this global
shifting mechanism, tile dynamics are quite similar to that of \emph{PushPush}.
In \cite{2048_pspace}, it is claimed\footnote{A brief discussion of the status of this result can be found in Appendix \ref{app1}.} that this game is PSPACE-complete by
a reduction from one form of Nondeterministic Constraint Logic. Such constructions
encode a regular directed planar graph with in-flow constraints with the objective
of flipping a certain edge or reaching a certain global configuration by flipping
edges without violating any in-flow constraints.

\section{Preliminaries}
We highlight some of the underlying assumptions and properties used throughout the construction.

\subsection{Board size} \label{basics:size}
Unlike the original game that uses a $4 \times 4$ board for \emph{2048} or \emph{Threes!}, we use boards of polynomial size. While we do not enforce square boards, this can always be achieved by padding a rectangular board, without blowing up the board size.

\subsection{No unpredictability} \label{basics:unpredictability}
In the original \emph{2048} game, after each successful move a single new tile having a value of either \emph{2} or \emph{4} is generated at a random empty cell. Our construction, on the other hand, will be fully determined by the initial board and no new tiles will be generated. In addition, the initial board will have no gaps, but gaps can be created as the result of merging tiles as mentioned in \ref{basics:lattice}.

\subsection{Checkerboard and Parity} \label{basics:lattice}
We use a \textit{static} background lattice to keep things under control. The construction guarantees that no merges can ever be performed using lattice tiles. To achieve this in \emph{2048}, we use a $2$-$4$ checkerboard lattice with all allowed merges being multiples of a $2 \times 2$ block, which effectively preserves parity and keeps the lattice intact.

\section{2048 Gadgets}\label{2048_gadgets}
We give a brief description of the gadgets we use in the construction. Note that all our gadgets use values greater than $4$ embedded in a special way inside the checkerboard lattice. Certain tiles in each gadget can be merged to communicate signals through the board by shifting other tiles inside other gadgets. We implemented our gadgets to allow the reader to try it for themselves. A brief description of the implementation and how to access it is outlined in Appendix \ref{app3}.

\subsection{Collapser Gadget}
This is the basic gadget we use to perform a single $2 \times 2$ shift, either horizontally or vertically, by merging into itself in a collapsing or folding manner. Once the merge is performed, we are left with a $2 \times 2$ checkerboard that cannot be merged any further. Figure \ref{fig:collapser} illustrates the collapser gadget, in both the horizontal and vertical configurations.

\begin{figure}[htb]
\includegraphics[width=0.5\textwidth]{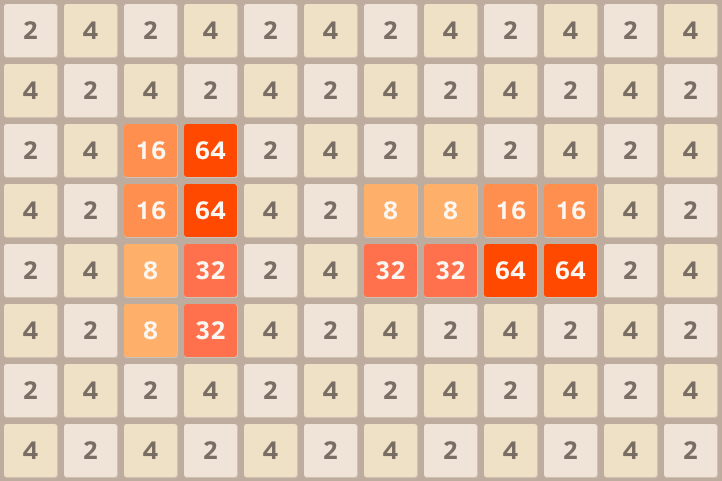}
\caption{Collapser Gadget: vertical (left) and horizontal (right) configurations.}\label{fig:collapser}
\end{figure}

\subsection{Variable Gadget}
The variable gadget is a horizontal sequence of tiles that encodes a variable assignment by 2 simultaneous $2 \times 2$ merges either to the right or left along the two rows containing the gadget. We force an \textit{activation sequence} on variables starting with $x_1$ at the bottom towards $x_n$ at the top. As we merge tiles inside the variable gadget, we also line up tiles in other columns which allows the activation of the next variable in the activation sequence and possibly satisfy clauses. We encode a \verb|TRUE| assignment by right merges and a \verb|FALSE| assignment by left merges. Note that regardless of the assignment of variable $x_i$, variable $x_{i+1}$ can still be activated by a $\downarrow$ move. This is achieved by a \textit{connector gadget} that allows a vertical shift following the horizontal shift corresponding to assigning a truth value to the variable. That vertical shift that follows can be used to activate the next variable in the sequence. In addition, variable $x_1$ comes activated in the initial board. Finally, note that the vertical separation between any two consequtive variables can be adjusted to any sufficiently large value.

\begin{figure}[htb]
\includegraphics[width=0.8\textwidth]{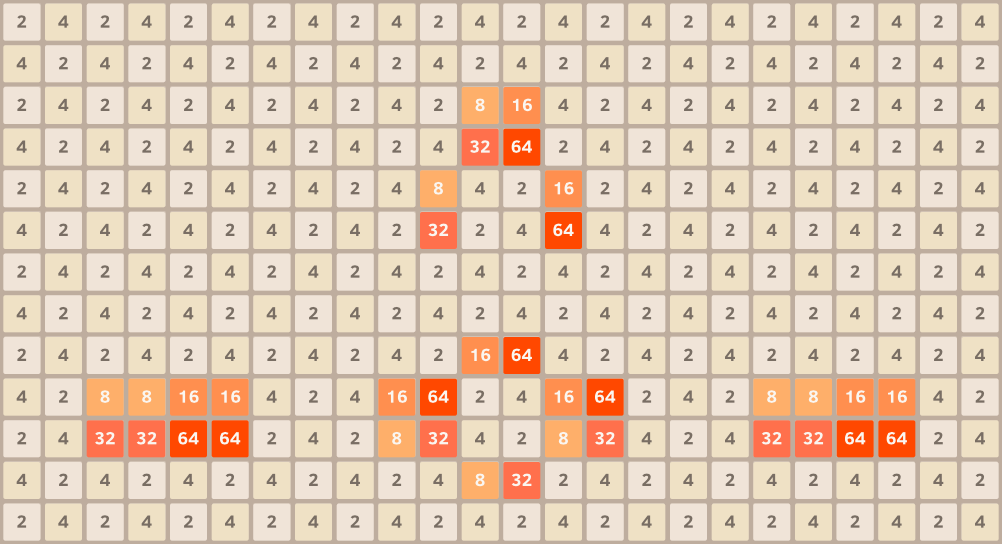}
\caption{Variable Gadget}\label{fig:variable_gadget}
\end{figure}

\subsection{Connector Gadget}\label{connector_gadget}
This gadget, shown in Figure \ref{fig:connector_gadget}, allows variable assignments to satisfy a clause. This bears similarity to the middle collection of tiles in the variable gadget, in Figure \ref{fig:variable_gadget}, and is used similarly to allow a vertical shift following the horizontal shift achieved by assigning a given variable. Note, however, that unlike the case for variable gadgets, where the middle part could be operated by both rightwards and leftwards shifts, this gadget needs to distinguish true assignments from false assignments. This is the reason why the $2 \times 2$ blocks alternate their \textit{sign} along each row.

The connector gadget has one row per variable and one column per clause. Satisfying a clause corresponds to being able to merge certain tiles along a given \textit{clause column}. This connection will be implemented in a manner similar to that of the middle portion of the variable gadget as in Figure \ref{fig:variable_gadget}, but in one of two ways corresponding to whether the literal is negated or not. By merging these tiles $\downarrow$, a special block at the top of the clause column is brought into alignment in a \textit{clause checker row}. Note that the connector gadget is only determined by $n$ and $m$. Furthermore, each literal can be correctly connected regardless of other literals.

\begin{figure}[htb]
\includegraphics[width=0.6\textwidth]{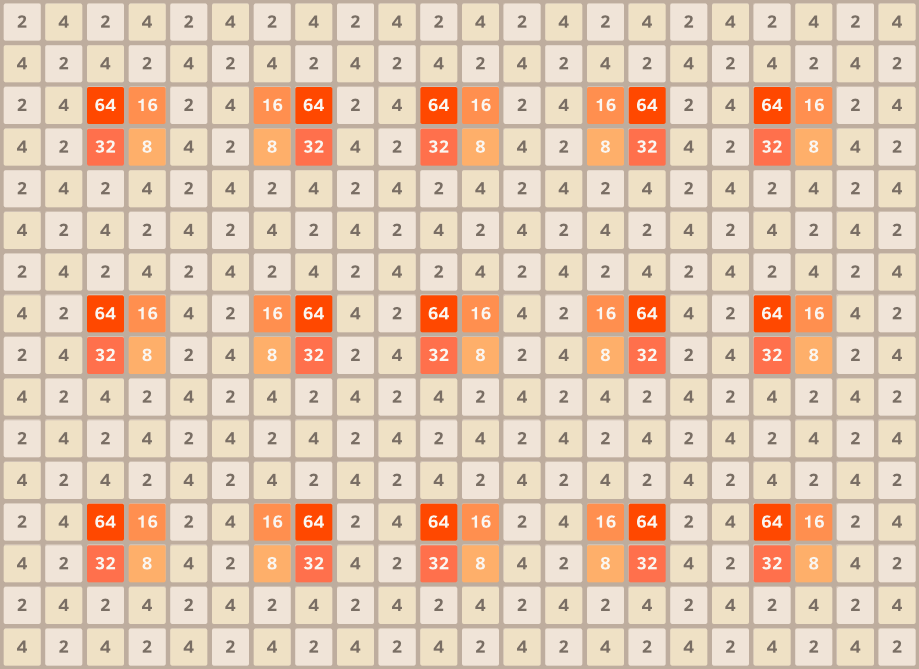}
\caption{Connector Gadget: Blocks alternate sign along each row.}\label{fig:connector_gadget}
\end{figure}

\subsection{Clause Gadget}\label{clause_gadget}
As shown in Figure \ref{fig:clause_gadget}, each clause will be represented by a column with a special block at the top and 3 literals connected in the connector gadget as described in \ref{connector_gadget}. Note that once the clause is satisfied by a correct merge along its column, no more merges in its column can be performed later by subsequent variables. This effectively \textit{leaves} the clause satisfied. Figure \ref{fig:clause_gadget_example} shows a larger example.

\begin{figure}[htb]
\includegraphics[width=0.6\textwidth]{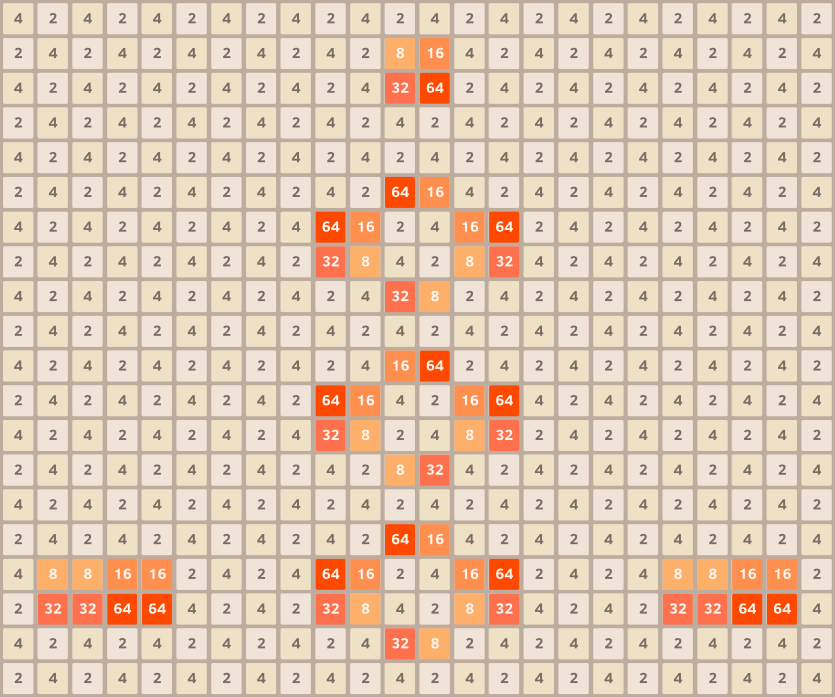}
\caption{Clause Gadget: The middle literal is negated in this example.}\label{fig:clause_gadget}
\end{figure}

\subsection{Key-Lock Gadget}
As it is more convenient to represent the verdict with a single event, we use a \textit{key-lock gadget} to force all satisfied clauses to wait, till all variable assignments are made, before they can be used for any further merges. The key-lock gadget is activated by a \textit{key}, which is just a \textit{special auxiliary variable}, that can only be activated as a result of assigning a truth value to the last variable $x_n$. Figure \ref{fig:key_lock_gadget} shows a key-lock gadget for 3 clauses. As we need to allow the clause column to pass uninterrupted through the lock gadget, we need to make two, rather than just one, $2 \times 2$ horizontal shifts to activate the lock. Note how lock activators interleave, by alternating signs. Finally, note that we actually need two $2 \times 2$ vertical shifts to activate the locks, as they cannot be on the same level as the special clause blocks in the initial configuration. This is why the gadget in Figure \ref{fig:key_lock_gadget} uses pairs of consequtive collapsers to activate the auxiliary variables, slide the keys into their locks and pull the locks down.

\begin{figure}[htb]
\includegraphics[width=0.8\textwidth]{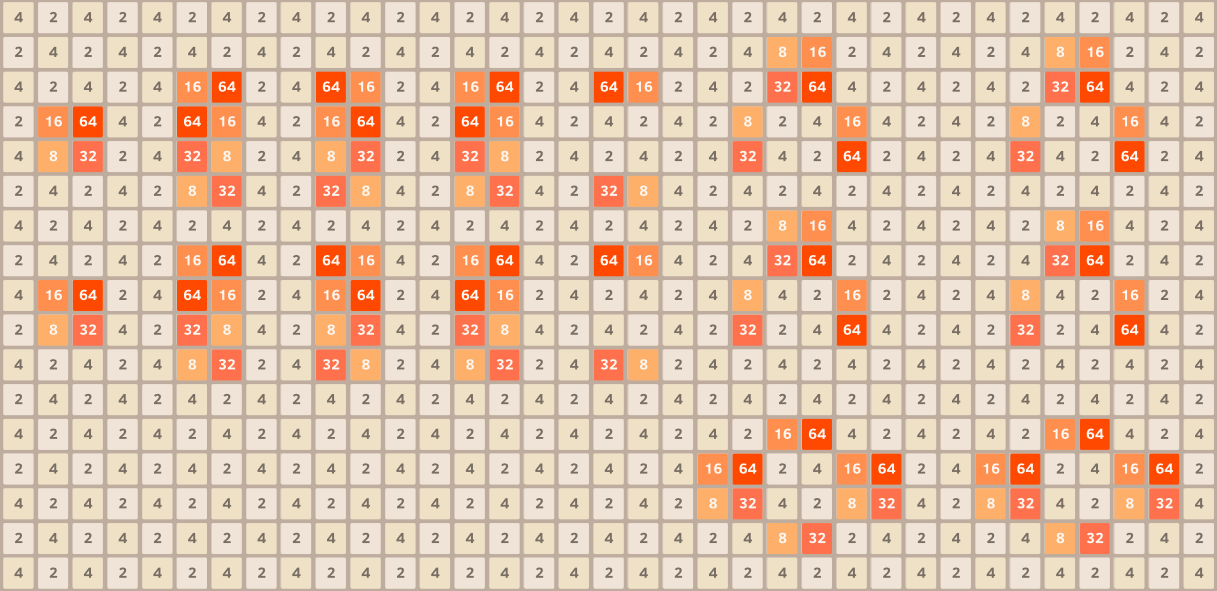}
\caption{Key-Lock Gadget for 3 clauses with the auxiliary variable to the right.}\label{fig:key_lock_gadget}
\end{figure}

\subsection{Checker Gadget}
The clause checker verifies that all clauses have been satisfied. As described in \ref{clause_gadget}, satisfying a clause enables a $2 \times 2$ downward merge that brings the special block at the top into what we refer to as the \textit{clause checker row}. This row allows the activation of a special \textit{goal collapser} when all clauses are satisfied.

\begin{figure}[htb]
\includegraphics[width=0.8\textwidth]{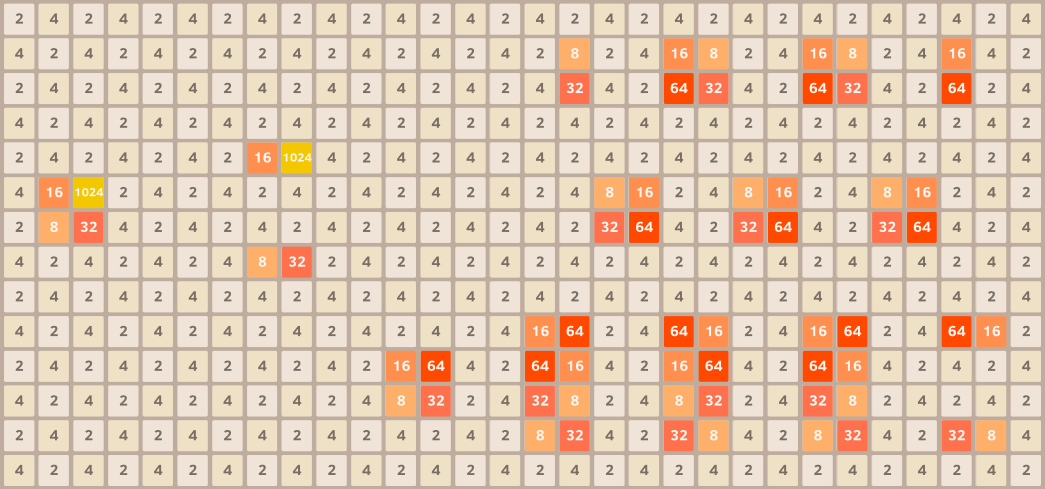}
\caption{Checker Gadget with all 3 clauses satisfied.}\label{fig:checker_gadget}
\end{figure}

\section{2048 Reduction}
In this section, we prove that the constructions described in Section \ref{2048_gadgets} work as intended.

\begin{observation}\label{obs:merges_2x2}
Merges can only occur as multiples of $2 \times 2$ merges.
\end{observation}
A matching pair of adjacent tiles are aligned either with a row or a column.  Notice that in the starting construction, for every matching pair there exists a second matching pair in the corresponding row (there are no vertical matching pairs in the starting configuration).  Thus, any merger that occurs will cause $2c$ merges, for some constanct $c$ (in this case, $c=2$).  Additionally, every matching pair has a second matching pair in an adjacent row.  Because of this, there will always be two rows which undergo the $2c$ merges described above, leading to every merger consisting of $2 \times 2 \times c$.  In other words, the existing $2 \times 8$ blocks of matching tiles will cause $c(2 \times 2)$ merges for each possible merger in the starting configuration.

The rest of the construction consists of similar $2 \times 8$ blocks of matching tiles, but they begin with the $2 \times 2$ tile block in their centers offset by two spaces.  This, in essence, renders each of these blocks inert until they are activated by shifting the center back into position.  When the centers are in position, we again get a structure as outlined above, leading to $c(2 \times 2)$ mergers. 

The construction is built so that each $2 \times 8$ block (barring those for the first variable) is inert until it has been activated by some other merger.  Every activation requires a $2 \times 2$ merger, every activated $2 \times 8$ block can cause a $2 \times 2$ merger, and these activations are chained together throughout the construction.  

\begin{lemma}
For any set of moves, the lattice is preserved.
\end{lemma}
\begin{proof}
The lattice consists of a tiling of a set of $2 \times 2$ tile blocks, with alternating values of $2$ and $4$.  So long as the edges of these $2 \times 2$ tiles are aligned, they will never merge.  Given Observation \ref{obs:merges_2x2}, all merges occur in multiples of $2$, meaning that these blocks will never be out of alignment, that they can never merge, and that the lattice will be preserved.
\end{proof}

\begin{observation}\label{obs:gap_boundary}
Gaps can only be created at the boundary.
\end{observation}
Any move made in the \verb|2048| game will, by definition, slide all tiles as far as possible in the selected direction.  Because of this, a gap will only ever be created at the end of a row or column.

\begin{fact}\label{cor:gaps_from_merges}
New gaps are only created when a merge occurs.
\end{fact}
Gaps on the boundary can be filled by sliding tiles, but in doing so another gap is created at the other end of the sequence of tiles.  For all intents and purposes, the gap has simply changed location.  To create a new gap (that is, to increase the number of gaps in the construction), tiles must merge together.

\begin{lemma}\label{lemma:one_merge}
No row or column can witness more than one merge event.
\end{lemma}
\begin{proof}
By construction, there are never $4$ matching tiles in a line, so there will never be the case that a merger can create two new adjacent matching tiles which will, in turn, merge.  Barring this, the only way to create a second merger event is to introduce new tiles into a row or column.  For these introduced tiles to have any effect requires a gap to exist somewhere other than on the boundary but, as shown in Observation \ref{obs:gap_boundary}, this can not occur.
\end{proof}

\begin{lemma}\label{lemma:gap_corners}
Gaps accumulate at the corners and, given a large enough padding, will not interfere with the construction.
\end{lemma}
\begin{proof}
Given \ref{obs:gap_boundary}, gaps will never occur within the active portion of the construction.  A gap on the boundary, then, can do one of two things: (i) a move perpendicular to the one that created the gap will fill it and create a new gap in a corner or (ii) a reversing move will fill the gap and create one on the opposite boundary.

For case (i), gaps in the corner can only interfere with the construction if they are large enough to share some rows or columns with the core containing all useful gadgets.  Given that no new tiles are added to the construction and with Lemma \ref{lemma:one_merge}, there is an easily calculated upper bound for the number of merges that may occur and thus a bound on the number of gaps created. We simply increase the padding around the core until it is greater than this bound.

For case (ii), alternating moves along the same direction have no effect until a perpendicular move is made.  No matter which side the gap is on, we end up back in case (i).
\end{proof}

\begin{lemma}
Effective game play alternates horizontal and vertical moves, in this particular order.
\end{lemma}
\begin{proof}
By effective game play we refer to a sequence of moves in which each move accomplishes something, specifically, each move either causes tiles to merge. For example, if the board configuration consists of only the lattice defined above, then there is no set of moves that would be considered effective game play as the lattice is locked in place.

To prove the lemma, first notice that in the initial configuration there are no gaps and the only merges that are possible are the ones involved in assigning a value to the first variable.  Thus, the very first move must be horizontal as this is required to make this assignment.

Given Lemma \ref{lemma:one_merge}, a repeated sequence of like moves (that is, all moves occurring either horizontally or vertically) will create no more than one merge event.  Because of this, a sequence of like moves has the same effect as the last move in the sequence and any such sequence in a set of moves can be reduced to just the last one.  When all such series of moves are reduced, one is left with alternating horizontal and vertical moves.
\end{proof}

\begin{lemma}\label{lemma:hv_move_lock}
Once the move sequence transitions from horizontal to vertical, or vice-versa, the effects of the previous decision cannot be unmade.
\end{lemma}
\begin{proof}
In order for a row or column to change it must either: (i) slide toward a boundary or (ii) have an internal merge event.

Case (i) will not occur as it requires a gap on the boundary of the row or column.  Corollary \ref{cor:gaps_from_merges} states that new gaps will be created on the boundary only when a merge event occurs.  Once an orthogonal move occurs, this gap will be plugged and further parallel moves will have no effect.  As each row or column can only witness a single merge event (see Lemma \ref{lemma:one_merge}), these gaps will never reoccur.

Case (ii) will not occur either.  Any previous decision that has been made for a row or column entails the occurrence of a merge to assign a value to a variable, determine whether a clause has been satisfied, etc.  But again, Lemma \ref{lemma:one_merge} does not allow for multiple mergers to happen within the same row or column.
\end{proof}

\begin{corollary}
Variables are assigned consistently and thus the satisfaction of clauses is valid.
\end{corollary}
\begin{proof}
By construction, clause gadgets will only be satisfied (that is, shift $\downarrow$ into place) once one of the variable rows has activated a block in the clause gadget's column.  This activation occurs only if a variable has already been assigned a value which leads to a satisfying literal for the clause.  Lemma \ref{lemma:hv_move_lock} states that once this assignment has occurred, the row representing the value for that variable can not be changed, meaning that the value of its literals can not change as they all exist in the same row as its variable. Finally, if one of the literal portions of a clause gadget allows that clause to shift $\downarrow$, then the literal must be \verb|TRUE|, the variable must be appropriately assigned, and the clause itself is truly satisfied by the current assignment.
\end{proof}

\begin{lemma}
In any winning game, every sequence of vertical moves must end in a $\downarrow$ move.
\end{lemma}
\begin{proof}
In order for the checker gadget to verify the satisfaction of all clauses, the upper most block of each clause gadget must have been pulled $\downarrow$ into place.  Without such a shift, the length of the checker will be interrupted by portions of the lattice and it will be unable to merge enough tiles to shift the goal collapser into place.
\end{proof}

\begin{lemma}\label{lemma:1024_align}
The two $1024$ tiles can be aligned if and only if all clauses are satisfied.
\end{lemma}
\begin{proof}
Let the two $1024$ be aligned at some point in our sequence of moves. By Lemma \ref{lemma:gap_corners}, the gaps created have no effect on the core of the construction. Hence the only way the two $1024$ tiles can meet, is by $2m$ merges in the row containing the left $1024$ block, where $m$ is the number of clauses in \verb|3-SAT|. This is only possible when all of the key-locks and clauses are aligned with each other. By construction this would mean all clause blocks have shifted down by $2\times2$, which implies at least one of the variable  in each clause is  assigned a value that satisfies the respective clause. \\
For the reverse suppose all of the clauses were satisfied. That would mean all the clause gadgets have moved down $2\times2$. We can activate all of the key-lock gadgets by assigning a value of \verb|TRUE| to the auxiliary variable, and shift them down by four rows so that they are aligned with the clause gadgets. A single $\rightarrow$  move would cause $2m$ merges and make the two $1024$ blocks adjacent to each other.
\end{proof}
\begin{theorem}\label{theorem:2048_nphard}
$2048$ is NP-hard.
\end{theorem}
\begin{proof}
By Lemma \ref{lemma:one_merge} each row and column can witness at most one merge event. Then by construction, the only way a tile of $2048$ is generated is by merging of the two $1024$ tiles. That implies we can generate a tile of value $2048$ if and only if the two $1024$ tiles are aligned with each other. By Lemma \ref{lemma:1024_align}, this would mean we can achieve a tile of value $2048$ if and only if all of the clauses are satisfied. This completes our reduction from \verb|3-SAT|, and proves \verb|2048-Game| is NP-hard.
\end{proof}

\begin{theorem}\label{theorem:2048_NP}
$2048 \in NP$ \cite{2048_np}.
\end{theorem}

From theorem \ref{theorem:2048_NP} and theorem \ref{theorem:2048_nphard} we have the following
\begin{corollary}
\verb|2048-Game| is NP-Complete
\end{corollary}

\section{The Case for \emph{Threes!}}
We propose a similar approach to proving the hardness of \emph{Threes!}. We generalize the game to a rectangular board of arbitrary size and disallow the generation of new tiles. Note that in the original \emph{Threes!} a tile of value \emph{1} or \emph{2} is generated in an empty cell. As for the background lattice, the situation is much simpler in \emph{Threes!} as \emph{1} tiles do not merge and we do not need to worry about parity.

In the remainder of this section, we design gadgets similar to the ones we used for \emph{2048} and discuss the difficulties that arise when dealing with persistent gaps. An implementation of these gadgets is also provided for the reader's convenience as outlined in Appendix \ref{app3}.
 
\subsection{Variable Gadget}
The variable gadget is a group of tiles with value $3$ in the same row as shown in the bottom half of Figure \ref{fig:Threes_variable_gadget}. We encode a \verb|TRUE|  assignment by a right merge and \verb|FALSE| assignment as a left merge. A right move merges the pair of $3$s in the same row. This brings the bottom $3$ in alignment with the middle $3$, in the connector gadget. Note that a left move would bring the middle $3$ in alignment with the top one, in the connector gadget. Hence regardless of the assignment of $x_1$ we can activate the next variable gadget $x_2$ (top half of Figure \ref{fig:Threes_variable_gadget}), by a single $\downarrow$ move, as two of the $3$ tiles in the connector gadget merge to form a single tile of value $6$, thereby bringing the tiles of $x_2$ in alignment.

\subsection{Clause Gadget}
Each clause is represented by a column with a special tile with a value of $6$ at the top. Every positive literal is represented by a pair of tiles with value  $3$, one of them being to the north west of the other, for example the middle pair of tiles in Figure \ref{fig:Threes_clause_gadget}. Every negative literal has one tile to the north east of the other. The figure gives an example of a clause $(\bar{x_1}\wedge{x_2}\wedge{\bar{x_3}})$. It's easy to see that once  a positive literal is assigned \verb|TRUE|, or if a negative literal is assigned a value of \verb|FALSE|, the tiles corresponding to the literal align up in the same column. Hence the $6$ at the top of each clause column can be brought one step down by performing a $\downarrow$ move and merging two of the tiles with value $3$. Such a merge corresponds to the satisfaction of a clause. Note that once a clause has been satisfied by suitable assignment of a literal, we do not care about aligning the tiles of the rest of the literals. Figure \ref{fig:Threes_clause_gadget_example} shows a larger example.

\begin{figure}[htb]
\includegraphics[width=0.8\textwidth]{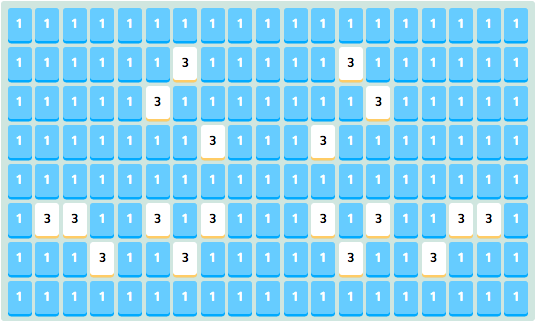}
\caption{Variable Gadget.}\label{fig:Threes_variable_gadget}
\end{figure}

\begin{figure}[htb]
\includegraphics[width=0.2\textwidth]{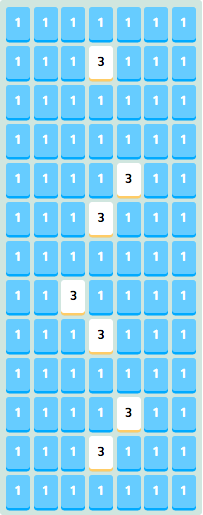}
\caption{Clause Gadget.}\label{fig:Threes_clause_gadget}
\end{figure}

\begin{figure}[htb]
\includegraphics[width=0.8\textwidth]{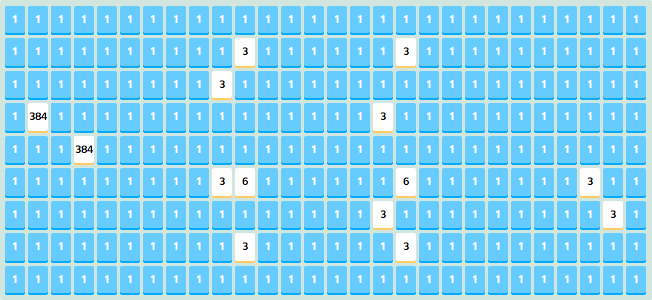}
\caption{Checker Gadget with 2 clauses. The right clause is satisfied.}\label{fig:Threes_checker_gadget}
\end{figure}

\subsection{\emph{Threes!} Reduction}

While in \emph{Threes!}, we can use a much simpler lattice, the dynamics allow gaps to persist. This could be problematic if it allows merging tiles in a way that achieves the goal where the original \verb|3-SAT| instance is not satisfiable. For example, such \text{non-canonical} merges could be used to satisfy more clauses without using their literal connections or allow a variable to act as both true and false. To remedy this situation, we spread apart clause gadgets such that the gaps created at the top row as clauses are satisfied are not allowed to travel to other clause gadgets. Other than that, the gadgets look pretty similar to the ones used for $2048$.

Figure \ref{fig:Threes_variable_gadget} shows the variable gadget. Unlike the one for $2048$, variable induced merges occur one at a time. Still, the gadget for \emph{Threes!} ensures a consistent assignment of variables, and the next variable is activated as in $2048$. The clause gadget is shown in Figure \ref{fig:Threes_clause_gadget}.

We use a similar key-lock gadget to verify that all clauses are satisfied. When this holds, special tiles line up to allow the creation of an arbitrary large value, i.e. greater than 12. This event corresponds to the solution to the decision problem at hand. Figure \ref{fig:Threes_checker_gadget} shows the checker gadget for \emph{Threes!} for both satisfied and unsatisfied clauses. Note that part of the key gadget is embedded in the clause column and only shifts down to align with its key when the clause is satisfied. Fully unlocking the gadget requires two downward moves, which is enabled by a 2-level collapser gadget, as shown in Figure \ref{fig:Threes_checker_gadget}.

With the properties of the reduction for $2048$ in mind, we only need to show a few more properties for this new reduction.

\begin{observation}
Gaps cannot travel far and may not interfere with the construction.
\end{observation}
\begin{proof}
Without loss of generality, consider a gap exists in the top most row which the player may attempt to move rightwards.  In order to do so, the player must shift the row to the right and in order to do that, there has to be another gap further down the same row.  In fact, for every step the gap moves to the right, there must be a gap further to the right into which we can slide the row.  

These gaps are only created when a merge occurs inside of the construction and, as there is an obvious bound on the number of merges possible (given a polynomial number of tiles) there is a bound on the distance any gap may travel. We can simply pad our gadgets, creating enough space between them so that no gap may travel from one gadget to another.  This would guarantee that no gaps created at one gadget may interfere with the operation of other gadgets.
\end{proof}

With that, the following results seem to be within reach.

\begin{conjecture}\label{theorem:Threes_nphard}
$THREES$ is $NP$-hard.
\end{conjecture}

\begin{conjecture}\label{theorem:Threes_NP}
$THREES \in NP$.
\end{conjecture}

\appendix

\section{A Note on the PSPACE-Completeness of \emph{2048}}\label{app1}
We spotted a number of issues in the constructions of \cite{2048_pspace}, which have been communicated to and confirmed by the authors. Namely, due to the irreversible nature of merging tile blocks used to encode directed edges of NCL graphs, the reduction is effectively creating bounded NCL games. Such games are only NP-Complete, not PSPACE-Complete as the paper was aiming. Aside from that, tiles are \textit{sometimes} assumed to move in a manner which is different from the original game and the activation sequences of different gates are not of the same length which would complicate things as each move is applied to the whole board.

The generalization we use in this paper is arguably simpler than the one used in \cite{2048_pspace} as we do not use any tiles of value higher than \emph{1024}. A more important distinction is that we disable the generation of new tiles while in \cite{2048_pspace} new tiles of specific values had to be generated at specific locations to make the reduction work.

\section{Playable Gadgets}\label{app3}
We provide playable versions of our gadgets by adapting open source implementations of the games. We use the codes created by \href{https://github.com/CyberZHG/2048}{CyberZHG} and \href{https://github.com/angelali/threesjs}{angelali}. The playable gadgets can be accessed on the following links: \href{http://cs.umd.edu/~akader/projects/alb/internal/index.html}{2048}, \href{http://cs.umd.edu/~akader/projects/alb/internal3/index.html}{Threes!}.

\section{Additional Figures}\label{app2}

\begin{figure}[htb]
\includegraphics[width=0.8\textwidth]{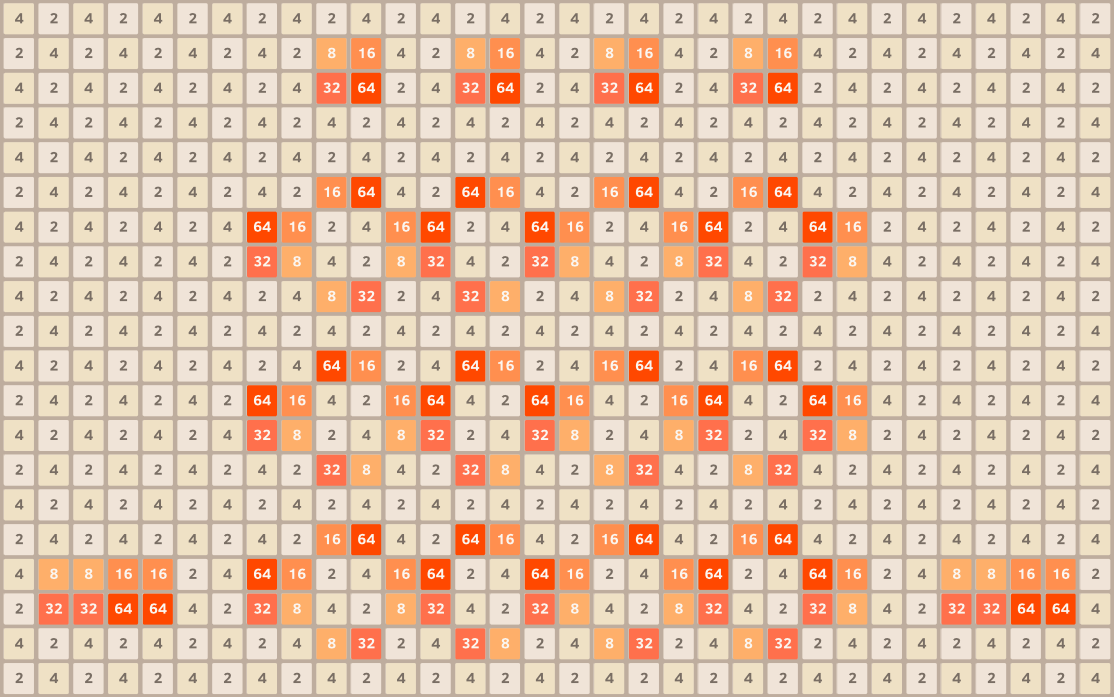}
    \caption{Example of the clause gadget for 2048 having 4 clauses.}\label{fig:clause_gadget_example}
\end{figure}

\begin{figure}[htb]
\includegraphics[width=0.8\textwidth]{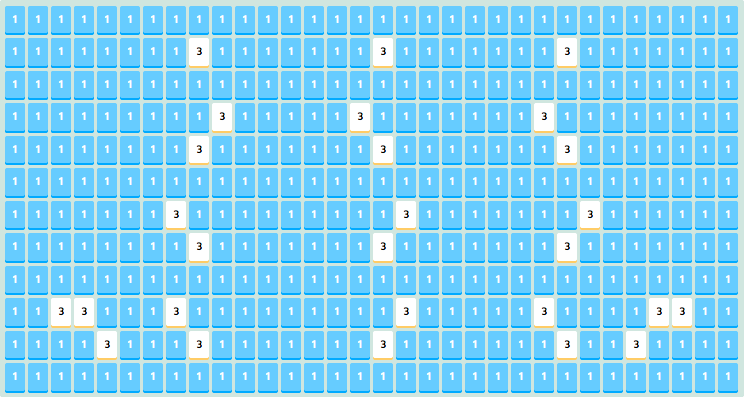}
\caption{Example of the clause gadget for \emph{Threes!} having 3 clauses.}\label{fig:Threes_clause_gadget_example}
\end{figure}

\bibliographystyle{amsplain}
\bibliography{refs}

\end{document}